\documentclass{article}
\usepackage{preamble}
\usepackage{authblk}

\title{\bf Quantum Algorithms and Lower Bounds for Eccentricity, Radius, and Diameter in Undirected Graphs}

\author{Adam Wesołowski \thanks{Department of Computer Science, Royal Holloway University of London. \href{mailto:adam.wesolowski.2023@live.rhul.ac.uk}{Email: adam.wesolowski.2023@live.rhul.ac.uk}.} \qquad \qquad \qquad Jinge Bao \thanks{School of Informatics, University of Edinburgh. \href{mailto:jinge.bao@ed.ac.uk}{Email: jinge.bao@ed.ac.uk}.}~~\thanks{Corresponding author.}}

\date{}
\begin{document}

\maketitle
\begin{abstract}
The problems of computing eccentricity, radius, and diameter are fundamental to graph theory. These parameters are intrinsically defined based on the distance metric of the graph. In this work, we propose quantum algorithms for the diameter and radius of undirected, weighted graphs in the adjacency list model. The algorithms output diameter and radius with the corresponding paths in $\widetilde{O}(n\sqrt{m})$ time. Additionally, for the diameter, we present a quantum algorithm that approximates the diameter within a $2/3$ ratio in $\widetilde{O}(\sqrt{m}n^{3/4})$ time. We also establish quantum query lower bounds of $\Omega(\sqrt{nm})$ for all the aforementioned problems through a reduction from the minima finding problem. 

\end{abstract}

% \tableofcontents

\section{Introduction}
Given an \textit{undirected, weighted, and non-self-loop} graph $G=(V, E, w)$, the \textit{Eccentricity} $Ecc$, \textit{Radius} $\mathcal{R}$ and \textit{Diameter} $\mathcal{D}$ of the graph $G$ are defined as follows
\[
    Ecc(u) = \max_{v \in V}d(u,v)
\]
\[
    \mathcal{D} = \max_{s,t \in V, s \neq t}d(s,t)
\]
\[
    \mathcal{R} = \max_{s,t \in V, s \neq t}d(s,t)
\]
where $u$ is a fixed node and $d(s,t)$ is the distance i.e. the length of the shortest path between two nodes $s$ and $t$.  
In classical algorithmic research, all three problems are well understood, both in exact and approximation versions. There exists a long line of research on diameter, radius, and eccentricity problems in both directed and undirected graphs~\cite{basch1995diameter,albert1999diameter,aingworth1999fast,roditty2013fast,chechik,backurs2018towards,dalirrooyfard2019approximation, bentert2023parameterized,BORASSI201559}. Nevertheless, in the area of quantum algorithms, we could not find any results apart from the works in the quantum CONGEST model~\cite{le2018sublinear,wu2022quantum}. It appears that, in a quantum setting, these graph problems are relatively understudied compared to other graph problems such as triangle detection~\cite{izumi2019quantum,le2014improved,magniez2007quantum}, subgraph finding~\cite{LEGALL2016569,lee2011learning, magniez2007quantum}, or shortest paths finding~\cite{durr2006quantum,jeffery2023quantum, wesolowski2024advances}.

Classically, it is well known and established that eccentricity can be computed in $O(m)$ time and both radius and diameter in $O(\min(nm,n^\omega))$ time, where $\omega$ is the complexity of matrix multiplication. For a long time, it has not been known whether both diameter and radius can be found faster than by solving the APSP problem and postprocessing the results. In~\cite{williams2018some}, the authors proved the equivalence between \texttt{Radius} and APSP problems; however, \texttt{Diameter} and APSP equivalence have not been shown yet, and it is possible that faster algorithms for diameter may exist. However, quantumly, the situation is fundamentally different. We show in this work that, whereas quantum APSP is conjectured to have a lower bound of $\Omega(n^{1.5}\sqrt{m})$~\cite{ambainis2022matching}, radius (and diameter) can be computed in $\wt{O}(n\sqrt{m})$, therefore likely breaking the radius-APSP equivalence seen in classical computation. Moreover, this discrepancy demonstrates the more intricate hierarchy inside the quantum APSP class~\cite{allcock2023quantum} and that the tight classical complexity relations do not carry over to the quantum setting. Interestingly, the quantum algorithms for diameter and radius in this work also surpass the performance of classical matrix multiplication algorithms~\cite{alon1992witnesses,seidel1995all,alon1997exponent,cygan2015algorithmic}, and it is not known whether matrix multiplication admits quantum speedup in a general case.

In this work, we also undertake the study of the quantum query lower bounds for the aforementioned problems. In the classical setting, radius admits the $\Omega(nm)$ lower bound under the APSP conjecture, and diameter admits the conditional bound of $\Omega(n^2)$ under the Strong Exponential Time Hypothesis (SETH)~\cite{roditty2013fast}. This discrepancy exists because it is not known whether diameter is actually equivalent to APSP. These lower bounds are unlikely to carry over to the quantum computational setting. In fact, the upper bounds we have obtained in this work strongly suggest that quantum lower bounds are going to differ. In order to shed more light on how quantum lower bounds differ from the classical bounds, we provide a discussion of reductions that yield worse lower bounds than the newly introduced lower bounds (which are part of this work's contribution). We do this to highlight the fact that some reductions that classically seem to provide good bounds may fail to do so in a quantum setting. We show that we can obtain an $\Omega(\sqrt{nm})$ query lower bound for all of the considered problems by reducing from quantum minima finding on $d$ items of different types~\cite{durr2006quantum}.

We also initiate a study of quantum approximation algorithms for diameter on \textit{undirected, weighted} graphs. Here, we provide a short outline of classical research on approximation algorithms for diameter; for a more comprehensive overview, we refer the reader to~\cite{backurs2018towards}. As computing the diameter of a graph faster than in $O(n^2)$ seems hard, the approximation becomes a good compromise between efficiency and correctness. Estimating the diameter within a ratio of $1/2$ can be realized by performing BFS for every vertex and outputting the depth of the BFS tree, denoted by $d$. The diameter $\mathcal{D}$ is between $d$ and $2d$. In~\cite{aingworth1999fast}, the authors presented an $O(m\sqrt{n\log n})$ algorithm for distinguishing graphs of diameter $2$ and $4$. This algorithm was extended to obtaining a ratio $2/3$ approximation to the diameter in time $O(m\sqrt{n\log n}+n^2\log n)$, and it can be generalized to the case of \textit{directed} graphs with \textit{arbitrary positive real weights} on the edges, which is the first combinatorial algorithm to approximate diameter within a better ratio than a trivial algorithm. A more precise analysis was given by~\cite{roditty2013fast}, which shows that the algorithm brings slightly more efficiency when considering \textit{sparse} graphs. Note that~\cite{aingworth1999fast} presented an algorithm to approximate the distance matrix of APSP with an additive error of $2$ in $O(n^{2.5}\sqrt{\log{n}})$ time based on a similar idea. This algorithm returns not only distances but also paths. Furthermore, they gave a slightly more efficient algorithm for approximating the diameter of \textit{sparse} graphs. This algorithm is applied to the case of \textit{unweighted, undirected} graphs, but it can be generalized to the case of \textit{undirected} graphs with \textit{small integer edge weights}.

The importance of efficient computation of the diameter in the analysis of networks has been recognized in~\cite{Watts1998CollectiveDO, albert1999diameter}. Both radius and diameter are descriptive properties of networks and can be used to model optimal logistics, by helping in determining the center of a network, i.e. the node for which the distance to every other node in the network is minimized. Similarly, algorithms for vertex eccentricity, diameter, and radius are used as subroutines to many more complex problems arising in network security analysis, and showcasing faster algorithms has a direct effect on a long list of applications.
Most importantly, however, the three problems are very fundamental to algorithmic research and graph theory, and due to that fact, the demonstration of better algorithms and complexity results is immensely valuable for its own sake. There are some other works discussing the diameter through different perspectives, For example, the parameterized complexity of diameter was studied in~\cite{bentert2023parameterized}. And in distributed computing, diameter is well studied both classicallly~\cite{peleg2012distributed} and quantumly~\cite{le2018sublinear,wu2022quantum}.

% \jnote{williams and karl bringmann}

\subsection{Our results}

We first formulate the problems of computing the eccentricity, diameter, and radius problems.

\begin{problem}[\texttt{Eccentricity}]
Given an undirected graph $G(E, V, w)$, and a vertex $v \in V$, compute the eccentricity $Ecc(v)$ and return the corresponding path.
\end{problem}
\begin{problem}[\texttt{Diameter}]
Given an undirected graph $G(E, V, w)$, compute the diameter $\mathcal{D}(G)$ and return the corresponding path.
\end{problem}

\begin{problem}[\texttt{Radius}]
Given an undirected graph $G(E, V, w)$, compute the radius $\mathcal{R}(G)$ and return the corresponding path.
\end{problem}

The way we define the above problems slightly differs from the formulations one may find in other sources, namely we require not only a positive, real-valued number corresponding to each parameter but we also require the witness, i.e. the path of the corresponding total weight.\footnote{The total weight can be thought of as the length of the path, and for unweighted cases, the total weight is indeed equal to the length of a path.} Our reason for doing that has to do with the involved methods and a general belief that without techniques based on matrix multiplication, one cannot get such numbers anyway without finding some forms of witnesses.

In this work, we give the upper bound for \texttt{Diameter} and \texttt{Radius} in the quantum computational paradigm. The results are given in the adjacency list access model\footnote{Another common graph access model is the adjacency matrix model. We point the readers to~\cite{durr2006quantum})}. Our quantum algorithm combines the quantum search~\cite{boyer1998tight} with quantum single source shortest paths (SSSP) algorithm~\cite{durr2006quantum}.

\begin{theorem}[Upper bounds of computing \texttt{Diameter} and \texttt{Radius}]\label{thm:ul_UPPER}
    Given a graph $G=(V, E, w)$ in the adjacency list model, there exists a quantum algorithm that returns diameter $\mathcal{D}$ and radius $\mathcal{R}$ with the corresponding path, in $\wt{O}(n\sqrt{m})$ time.
\end{theorem}

As for the lower bound, we give the quantum query lower bound for solving \texttt{Eccentricity}, \texttt{Diameter}, and \texttt{Radius}. Our reduction is via combinatorial arguments by reducing from quantum minima finding of different types~\cite{durr2006quantum}.
% Importantly, we consider both exact and approxiamte algorithms for the problems under consideration. 
% In classical algorithmic research, all of the problems are well understood, both in exact and approximation versions. Nevertheless, in quantum algorithms research we could not find any results and it seems that these graph problems are relatively understudied compared to other graph problems such as triangle detection~\cite{magniez2007quantum, le2014improved}, subgraph finding~\cite{magniez2007quantum, lee2011learning}, cycle finding, single source shortest path, single pair shortest paths~\cite{durr2006quantum, cade2016time, wesolowski2024advances}. 
% For the lower bounds we provide a discussion of reductions which yield worse lower bounds than the newly introduced lower bounds (which are a part of this works contribution). We do that to highlight the fact that some reductions that classically seem to provide good bounds may fail to do so in a quantum setting. As an example, we quantize a CNF-SAT to Diameter reduction from \cite{roditty2013fast} under SETH\cite{IMPAGLIAZZO2001367}, and show that the bound obtained form hardness under QSETH\cite{buhrman2019quantum} matches 2 other different reductions, which all yield $\Omega(n)$, whereas we show we can obtain a $\Omega(\sqrt{nm})$ lower bound by reducing from Minima Finding problem on d items of different type\cite{durr2006quantum}.

\begin{theorem}[Lower bounds of computing \texttt{Eccentricity}, \texttt{Diameter} and \texttt{Radius} ]\label{thm:ul_LOWER}
    Given a graph $G=(V, E)$ in the adjacency list model, all quantum algorithms, that solve any of \texttt{Eccentricity}, \texttt{Diameter} and \texttt{Radius} require $\Omega(\sqrt{nm})$ queries.
\end{theorem}

As the diameter of a graph is hard to compute exactly, we propose the first quantum algorithm to approximate the diameter in the adjacency list model. 
We first formulate the problem as follows

\begin{problem}[$\varepsilon$-\texttt{Approximating diameter}]
Given an undirected graph $G(E, V, w)$, $w: E \mapsto \mathbb{R}^+$, output an estimate $\wh{D}$ of the diameter $\mathcal{D}(G)$, where $\varepsilon\mathcal{D} \leq \wh{D} \leq \mathcal{D}$ and the corresponding path of the diameter $\mathcal{D}$.
\end{problem}

Our quantum algorithm is inspired by~\cite{aingworth1999fast,roditty2013fast}. In their works, authors leverage the Partial Breadth-First Search as a subroutine. By combinatorial arguments, they proved that the longest distance from the nodes of any $s$-dominating set is a $2/3$ approximation of the diameter of the graph with high probability, i.e. $\varepsilon=2/3$. Instead of simply quantizing the classical algorithm, we refine the analysis of partial breadth-first search in the quantum setting by leveraging the quantum threshold searching~\cite{ambainis2004quantum}. The quantum partial BFS as a useful subroutine can be of independent interest

\begin{theorem}[Upper bound of approximating diameter]
    Given a graph $G=(V, E)$ in the adjacency list model, there exists a quantum algorithm that returns an estimate $\wh{D}$ of the diameter $\mathcal{D}$, where $2/3\mathcal{D} \leq \wh{\mathcal{D}} \leq \mathcal{D}$, with the corresponding path, in $\wt{O}(m^{1/2}n^{3/4})$ time. 
\end{theorem}

%\subsection{Other related or unrelated works. \textbf{(do we need this section?)}}
% [ABCP93] and [Coh93] presented efficient algorithm for computing $t$-strech paths for $t \geq 4$. [Coh94] presented an efficient algorithm that approximates paths from $s$ source to all other nodes in a \textit{weighted} graph in time $O((m+sn)n^\epsilon)$ for any $\epsilon >0$. This algorithm outputs paths of length

% \subsection{Main techniques}
% For approximating the diameter, our technique is inspired by~\cite{aingworth1999fast,roditty2013fast}. In their works, authors leverage the Partial Breadth First as the main subroutines. By combinatorical arguments, they proved that the shortest path from the nodes from the $s$-nearest neighbors is a $2/3$ approximation of the diameter of the graph. For the exact computation, we basically devise a nested quantum search algorithm\cite{boyer1998tight} with quantum SSSP algorithm~\cite{durr2006quantum}. The lower bounds are obtained by introducing a new reduction from quantum minima finding~\cite{durr2006quantum}.

% \subsection{Related works}
% Bentert and Nichterlein explore diameter by parameterized complexity way. 
% Yao and Bringmann considered approximating diameter in CONGEST model.

% Recently, \cite{alon2023sublinear} gave a sublinear algorithm to compute the shortest path in an expander graph.

\section{Preliminaries}

We use $G=(V,E,w)$ to denote a directed weighted graph, where $E \subseteq V \times V$ and $w: V \times V \mapsto \mathbb{Z}^+\cup\{\infty\}$. The diameter of $G$ is denoted by $\mathcal{D}(G)$, sometimes by ${\mathcal D}$ when there is no ambiguity. When $G$ is an unweighted graph, we can simplify the notation by $G=(V, E)$, where explicitly $w(u,v)=1$ when $(u,v) \in E$ or $w(u,v)=0$ otherwise, for each $(u,v) \in V \times V$. $N_s(v)$ for $v \in V$ is the $s$ closest neighbors of $v$ in $G$, i.e. the first $s$ nodes visited by performing BFS from $v$ in $G$. The traversal path is a $s$-partial BFS tree denoted by $\BFS_s(v)$, whose depth is $h_s(v)$. Moreover, the whole BFS tree from $v$ is $\BFS(v)$ of depth $h(v)$. We also use $\BFS(v,h)$ to denote the first $h$ levels of the BFS tree from $v$.

\subsection{Graph access model}

    % \par \textbf{Matrix Model.}
    % Given graph $G=(V,E,w)$, matrix model is given $i,j$ and output $w(i,j)$ in 1 query and $O(1)$ time. In quantum terminology, there is an oracle $O_G$ s.t.
    % $$
    %     \mathcal{O}_G^{mat}: \ket{i,j,0} \mapsto \ket{i,j,w(i,j)}
    % $$
    % \par \textbf{List Model.} Given a graph $G=(V,E,w)$, for each vertex $v \in V$, we can query the degree $deg(v)$ and explicitly there is a mapping $f_v: [deg(v)] \mapsto V$. In other way, there is an oracle $O_G^{lin}$ s.t.
    % $$
    %     \mathcal{O}_G^{lin}: \ket{u,i,0} \mapsto
    %     \ket{u,i,w(u,f(i))}
    % $$
    % And each query costs $O(1)$ time.
    % Given graph $G=(V, E)$, there are two common graph access models. One is the \emph{adjacency matrix model} and the other is the \emph{adjacency list model}.\\
    % \par \textbf{The adjacency matrix model}. The graph is given as the adjacency matrix $M_G \in \{0, 1\}^{n \times n}$, with $M_G[i,j] = 1$ if and only if $(v_i, v_j) \in E$. Given graph $G=(V,E,w)$, matrix model is given $i,j$ and output $w(i,j)$ in 1 query and $O(\poly\log{n})$ time. In quantum terminology, there is an oracle $\mathcal{O}_G^{mat}$ s.t.
    % $$
    %     \mathcal{O}_G^{mat}: \ket{i,j,0} \mapsto \ket{i,j,w(i,j)}
    % $$In the quantum setting, we assume access to an oracle that encodes the adjacency list information into quantum states. The oracle \( O_{\text{adj}}^G \) acts as follows:

    In this work, we assume the graph can be accessed in the adjacency list model. More specifically, to access a graph $G=(V, E, w)$, we are given the degrees $d_1,\ldots,d_n$ of each vertex, and for each vertex $v_i$, there exists an array with its neighbors $f_{v_i}: [d_i] \mapsto [n]$. In another way, the function $f_{v_i}(j)$ returns the weight of the $j$th neighbor of vertex $v_i$, according to some fixed numbering of the outgoing edges of $v_i$.
    Quantumly, we can formulate the following oracle there is an oracle:  
    \[
        \mathcal{O}_G : |v, i, 0,    0\rangle \mapsto |v, i, f_v(i), w(v, f_v(i))\rangle
    \]
for any vertex $v \in V$ and index $i \in [d_u]$, where $f_v(i)$ returns the $i$th neighbor of $v$, $w(u, f_v(i))$ is the weight of the corresponding edge. 

To analyze the running time of our quantum algorithm. Each of the single and two-qubit gates counts as one unit step. And we will use the QRAG (quantum random access gate) model as our quantum memory model~\cite{ambainis2007quantum}. Compared to the QRAM model (Quantum random access memory), QRAG can not only do ``quantum reading'' but also ``quantum writing''. Formally, $\mathsf{QRAG}$ gate works as
\[
        \mathsf{QRAG}: \ket{i}\ket{e}\ket{x} \mapsto \ket{i}\ket{x_i}\ket{x_1,\dots,x_{i-1},e,x_{i+1}\dots, x_N}
\]        
where
$i\in[N]$ and $e,x_1,\dots,x_{N} \in\{0,1\}^r$. And we assume that each memory access gate takes $O(1)$ time.

In our approximating algorithm, we need to use the quantum history-independent data structure designed by Ambainis~\cite{ambainis2007quantum}. With this quantum history-independent data structure, we are able to realize searching, insertion, and deletion in $O(\polylog(n))$ time.

\subsection{Quantum Subroutines}
We use the following generalized quantum search algorithm inherited from Grover search~\cite{grover1996fast}.
\begin{theorem}[Quantum search, $\mathtt{QSearch}(f)$~{\cite[Theorem 3]{boyer1998tight}}]\label{thm:qsearch}
    Given oracle access to a Boolean function $f : [N] \mapsto \{0,1\}$, such that the set of marked elements $M = \{x \in [N]: f(x) = 1\}$ has unknown size $k = |M|$, the \cref{alg:qsearch} finds a solution if there is one using an expected number of $O(\sqrt{\frac{N}{k}})$ Grover iterations. In the case that there is no solution, then QSearch terminates in $O(\sqrt{N})$ time.
\end{theorem}
% \jnote{The statement of this theorem is wrong}

\begin{algorithm}[htbp]
    
\caption{Quantum search~\cite{cade2016time}}
\label{alg:qsearch}
\begin{algorithmic}[1]

\Statex\textbf{Input:} $M$ elements with $d$ (unknown) marked elements 
\Statex \textbf{Output:} Marked elements
       \State Initialize $M = 1$ and set $\lambda = 6/5$.;
       \State Choose $j$ uniformly at random from the nonnegative integers smaller than $m$.
       \State Apply $j$ iterations of Grover’s algorithm, starting from initial state $\ket{\Psi_0} = \sum_i \frac{1}{\sqrt{N}}\ket{i}$.
       \item Observe the register: let $i$ be the outcome.
       \State If $i$ is indeed a solution, then the problem is solved: exit.
       \State  Otherwise, set $M$ to $\min(\lambda M, \sqrt{N})$ and go back to step 2.
   
   \end{algorithmic}
\end{algorithm}

We will use the following quantum subroutines frequently. The first one is Quantum Minimum Finding (QMF) ~\cite{durr1996quantum}. And the second one is called quantum counting or quantum threshold finding~\cite{ambainis2004quantum, buhrman1999bounds}. The third one is the quantum bread-first search algorithm to return the sequence of the visited nodes and the depth of the BFS tree. The last subroutine is the quantum algorithm to compute single source shortest paths (SSSP), which can be regarded as the quantum analogue of the classical Dijkstra algorithm~\cite{durr2006quantum}.

\begin{theorem}[Quantum minimum finding, $\mathtt{QMF}(f)$~\cite{durr1996quantum}]\label{thm:qmf}
    Given oracle access to a function $f:[N] \mapsto \R$, 
    % and the oracle $\mathcal{O}_f$ with access to function $f$, 
    there exists a quantum algorithm $\mathtt{QMF}(f)$ that finds $k$ minimal elements of different type from a set of $N$ elements with values in $\R$, with high probability and run time $\tilde{O}(\sqrt{kN})$.
\end{theorem}

\begin{theorem}[Quantum threshold finding, $\mathtt{QTF}(f,t)$~\cite{buhrman1999bounds}]\label{thm:qtf}
    Given oracle access to a function $f:[N] \mapsto \{0,1\}$ 
    % and the oracle $\mathcal{O}_f$ with access to function $f$ 
    and $t \in \mathbb{N}$ with $1 \leq t \leq N$ , there exists a quantum algorithm $\mathtt{QTF}(f,t)$ such that
    \begin{enumerate}
        \item if $|x| \leq t$ then the algorithm reports TRUE and outputs $x$ with certainty, and
        \item if $|x|>t$ then the algorithm reports FALSE with probability at least $9/10$.
    \end{enumerate}
    The algorithm makes $O(\sqrt{tN})$ queries to $x$ and has time complexity $O(\sqrt{tN}\log{N})$. 
\end{theorem}

\begin{theorem}[Quantum breadth-first search, $\mathtt{QBFS}(G,v)$~\cite{durr2006quantum}]
\label{thm:qbfs}
    Given an unweighted graph $G$, if the graph $G$ is connected, the algorithm $\mathtt{QBFS}(G,v)$ returns the set of all the nodes of $G$ and returns the depth of BFS tree started from $v$ in $O(n\log n)$ time in the adjacency list model.
\end{theorem}

% \begin{theorem}[Quantum partial breadth-first search $\mathtt{QPBFS}(G,v_0,s)$]
%     Given an unweighted graph $G=(V, E)$ a starting vertex $v \in V$ and partial number $s$, the algorithm $\mathtt{QPBFS}(G, v, s)$ (\cref{alg:q-partial-BFS}) can return the set of $s$-closest nodes to $v$ in the graph $G$ and the depth of $s$-partial BFS tree rooted at $v$ in $O(s^{3/2}\log^2 s)$ time.
% \end{theorem}

% We will use the following quantum subroutines frequently. The first one is Quantum Minimum Finding (QMF) \cite{durr1996quantum}. The second subroutine is the quantum algorithm to compute single source shortest paths (SSSP), which can be regarded as the quantum analogue of the classical Dijkstra algorithm~\cite{durr2006quantum}.

% \begin{theorem}[Quantum Minimum Finding~\cite{durr1996quantum}]
% \label{alg:qmf}
%     Let $a_1,...,a_n$ be $n$ integers, which can be accessed by a procedure $P$. There exists an quantum algorithm denoted by $\mathcal{A}
%     _{QMF}(\{a_i\}_i^n)$ that finds $\arg\min_i\{a_i\}_1^n$ and $\min_i\{a_i\}_1^n$ with success probability at least $2/3$ after applying $O(\sqrt{n})$ applications of $\mathcal{P}$.
% \end{theorem}

\begin{theorem}[Quantum single source shortest paths, $\mathtt{QSSSP}(G,v)$~\cite{durr2006quantum}]
\label{thm:alg-qsssp}
    Given an undirected graph of $G=(V, E, w)$ and a fixed node $v$, there exists a quantum algorithm denoted by $\mathtt{QSSSP}(G,v)$ that outputs the shortest path from node $v$ in $O(\sqrt{mn}\log^{5/2}{n})$ time with high probability. 
    % When the graph is given in the matrix model, the running time is $O(n^{3/2}\log{n})$. 
\end{theorem}

% \begin{corollary}[Quantum breadth-first search]
% \label{cor:alg-qbfs}
%     Given the array model of an undirected weighted (or unweighted) graph of $G=(V,E)$ and a fixed node $v_0$, there exists an quantum algorithm denoted by $\mathcal{A}_{qBFS}(G,v_0)$ that output a sequence of nodes and the depth of BFS tree in $O(\sqrt{mn}\log{n})$ time with high probability. 
%     % When graph is given in the matrix model, the running time is $O(n^{3/2})$. 
% \end{corollary}
% \jnote{belgi's papers}

% To introduce the important ingrediant which results into best quantum algorithm to approximate diameter with $3/2$ ratio, we use a result from "small counting loophole" of quantum algorithm. We restate the theorem here

% \begin{theorem}[Quantum threshold finding~\cite[Theorem 3]{buhrman1999bounds}]
% \label{q-threshold-finding}
% Given $t,N \in \mathbb{N}$ with $1 \leq t \leq N$ and oracle access to $x \in \{0,1\}^N$, there is a quantum algorithm such that
% \begin{enumerate}
%     \item if $|x| \leq t$ then the algorithm outputs $x$ with certainty, and
%     \item if $|x|>t$ then the algorithm reports so with probability at least $9/10$.
% \end{enumerate}
% The algorithm makes $O(\sqrt{tN})$ queries to $x$ and has time complexity $O(\sqrt{tN}\log{N})$.   
% \end{theorem}

% As we will use compostion of subrountine frequently, we need to introduce Quantum Variable Time Search Algorithm.

% \begin{theorem}[Quantum Variable Time Search]
% \label{thm:q_var_time_search}
% \end{theorem}

% \begin{corollary}[Quantum Variable Time Minimum Finding]
% \label{cor:q_var_time_min_find}
% \end{corollary}

\section{Computing of diameter, radius and eccentricity}\label{sec:computing}

\subsection{Upper bounds}\label{ssec:comp_up}

From \cref{table:results} below, one can see a pattern in complexity upper and lower bounds for problems based on finding shortest paths. The interesting open question that can be posed is whether the lower bounds for the problems admit the same ladder as the upper bounds do. The fact that up to polylogarithmic factors the gap for \texttt{Eccentricity} is essentially closed may lead one to believe that the \texttt{Diameter}/\texttt{Radius} quantum lower bounds should be slightly higher than the $\Omega(\sqrt{nm})$ achieved in this work, as it is reasonable to believe that these problems are slightly harder than \texttt{Eccentricity}. 
\begin{table}[ht]
\centering
\begin{tabular}{|c c c c c|}
\hline
&Single Pair Shortest Path& Eccentricity & Diameter/Radius& APSP\\
\hline
Upper bound& $\wt{O}(\sqrt{lm})$~\cite{wesolowski2024advances} & $\wt{O}(\sqrt{nm})$&$\wt{O}(n\sqrt{m})$ & $\wt{O}(n^{1.5}\sqrt{m})$ \\
Lower bound & $\Omega(\sqrt{lm})$& $\Omega(\sqrt{nm})$& $\Omega(\sqrt{nm})$& $\Omega(n^{1.5}\sqrt{m})(conjecture)$\\
\hline
\end{tabular}
\caption{{\bf The quantum time complexity upper bounds for graph problems based on shortest path finding and query lower bounds.} The $\Omega(\sqrt{lm})$ lower bound for SPSP problem is a direct corollary of the contributions outlined in this work. The upper bound for complexity corresponds to a trivial algorithm based on the quantum SSSP algorithm. The upper bound for \texttt{Diameter} and \texttt{Radius} is proved in this work. The upper bound for APSP is obtained by running the SSSP algorithm from every node and simple postprocessing. The stated lower bound of $\Omega(n^{1.5}\sqrt{m})$for APSP is conjectured to hold but not proven\cite{ambainis2022matching}.}\label{table:results}
\end{table}

The straightforward idea is to run quantum SSSP~\cite{durr2006quantum} for every node and then search over $O(n^2)$ shortest paths for the longest, or the shortest among the $n$ eccentricities paths, in both cases effectively solving the APSP problem. Solving the APSP problem is a well-known approach to finding the extremal eccentricities in a graph, i.e. the radius and diameter. It was duly noted in ~\cite{roditty2013fast} that finding all paths in a graph to output a single number seems excessive and perhaps unwarranted. Matrix multiplication algorithm does provide lower complexity w.r.t other classical approaches, and in a sense it does incorporate the aforementioned observation by not providing a witness, i.e. a path to the number it outputs as a solution. The APSP approach to solving \texttt{Diameter} and \texttt{Radius} outputs a set of $n(n-1)/2$ shortest paths and searching over these for the longest one requires only $O(n)$ steps via a version of the QSearch algorithm. The time complexity of the algorithm via APSP approach is $O(n^{\frac{3}{2}}\sqrt{m}\log^2{n})$ which is lower than the time complexity of the best classical matrix multiplication algorithm $O(n^{\omega}\log{n})$, when $\sqrt{m}<n^{0.87286}$. 
However, as one of the main contributions of this work, we give a quantum algorithm that has complexity significantly lower than that, requiring only $\wt{O}(n\sqrt{m})$ steps to output the diameter or radius and the corresponding path.
The significant advantage comes from offloading a sizeable amount of computation to the oracle of the $\mathtt{QSearch}$ algorithm. By doing so we reduce the search space from $O(n^2)$ to $O(n)$ at the cost of $O(\sqrt{n})$ additive runtime in the time complexity of running the oracle.

Instead of computing single source shortest paths with quantum SSSP algorithm for every node in a graph and then searching for the ``right'' path, our search is done over $n$ vertices. 
The oracle given a query consisting of a vertex $v$ runs the quantum SSSP from $v$ and subsequently searches for the longest path across all paths from $v$ to all other vertices. It outputs the length of the longest path found. 
The problem of diameter/radius finding comes down to finding the maximum/minimum value in a database of $n$ values. The oracle accepts the vertex with the largest/smallest eccentricity and outputs a vertex $v$. 
At this point, we are only given an assurance that the eccentricity of $v$ corresponds to the diameter/radius. Therefore we need to run quantum SSSP again for the found vertex $v$, and again find the path of diameter/radius length.

Via pushing one layer of search into the oracle we can lower the complexity to the nontrivial\footnote{Nontrivial here means, that the algorithm does not at any point solve the APSP problem. } $\wt{O}(n\sqrt{m})$ runtime. We use notation $\mathtt{QSearch}_{min}$, $\mathtt{QSearch}_{max}$
referring to the versions of the $\mathtt{QSearch}$ algorithm that accept the solution(s) corresponding to the radius (minimum eccentricity across vertices $v \in V$) and diameter (maximum eccentricity across vertices $v \in V$), respectively.

\begin{algorithm}[htbp]
\caption{Quantum algorithm for finding the diameter/radius of a graph $G$}\label{alg:cap}
\hspace*{\algorithmicindent} \textbf{Input:} Adjacency list of graph $G(V, E, w)$\\
\hspace*{\algorithmicindent} \textbf{Output:} Diameter (or radius) of $G(V, E, w)$ 
\begin{algorithmic}[1]
    \State Perform the $\mathtt{QSearch}_{max/min}$ over the $n$ vertices (for each vertex querying the oracle) accepting the vertex $v^*$ with maximal eccentricity (for the diameter) and minimal eccentricity (for the radius). Save the vertex $v^*$.
    
    \State (oracle) for a vertex in $G$ run the quantum SSSP algorithm. Store the $n-1$ found paths.
    
    \State (oracle) among the $n-1$ paths find the longest path and store its length, deleting all paths from memory. Output the length of the path (eccentricity of $v$).
    
    \State From the vertex $v^*$ run quantum SSSP and find the path corresponding to eccentricity of $v^*$. Output the eccentricty and the path.

\end{algorithmic}
\end{algorithm}

\begin{proof}{(\cref{thm:ul_UPPER})}.
The \cref{alg:cap} can return the diameter or radius with high probability, the correction of the algorithm follows~\cite {durr1996quantum,durr2006quantum}. It may be the case that there is more than one pair of vertices with the path of length corresponding to either diameter or radius. 

Searching over $n$ elements with $d$  ``good'' elements takes $O(\sqrt{\frac{n}{d}})$ as assured by~\cref{thm:qsearch}. Thus the cost of~\cref{alg:cap} is $O(\sqrt{\frac{n}{d}}\cdot \gamma)$ where $\gamma$ is the cost of a singular run of the oracle. The single oracle evaluation consists of running the quantum SSSP algorithm followed by a search over its outputs (i.e. $n-1$ shortest paths to all other vertices) for the length of the longest path. The cost of quantum SSSP algorithm is $O(\sqrt{nm}\log^{\frac{3}{2}}{n})$~\cite{durr2006quantum}, and the cost of search across $n-1$ paths for the longest one (assuming there is only one longest path in the worst case) is given by the complexity of the standard Grover search technique $O(\sqrt{n})$. The total cost of a single oracle iteration is $O(\sqrt{nm}\log^{\frac{3}{2}}{n}+ \sqrt{n})$. The last step after the search outputs a vertex whose eccentricity corresponds to the diameter or the radius, we need one more iteration of the quantum SSSP algorithm to find all shortest paths from $v$ and the search over these $n-1$ paths for the longest one, this an additive cost of $O(\sqrt{nm}\log^{\frac{3}{2}}{n}+ \sqrt{n})$. Thus the total time complexity of the diameter/radius finding algorithm is $O(\sqrt{n}(\sqrt{nm}\log^{\frac{3}{2}}{n}+\sqrt{n})+\sqrt{nm}\log^{\frac{3}{2}}{n}+ \sqrt{n})=\wt{O}(n\sqrt{m})$, as claimed. 
\end{proof}

The runtime of the quantum algorithm described in this section is sub-quadratic for all but maximally dense graphs (graphs where $m=O(n^2)$). Subquadratic runtime is unachievable for exact classical computation and classical approximation algorithms are only able to guarantee it for sparse graphs~\cite{chechik,roditty2013fast}. The main question of interest is whether there exist quantum approximation algorithms running in even lower complexity, which we discuss in subsequent sections. 

\cref{alg:cap} presented in this section outputs exact solutions to the \texttt{Diameter} ad \texttt{Radius} problems. In sparse graphs i.e when $m=O(n)$ it runs in $O(n^{1.5})$ steps, for dense graphs the runtime of the outlined algorithm is $O(n^{2})$, where the best classical exact computation requires $O(n^{\omega})$ steps for dense graphs and $O(n^2)$ for sparse graphs. It is a classically established result, observed in~\cite{roditty2013fast} that in sparse graphs even 3/2-approximating the diameter requires the same time as solving the exact version of APSP. This observation does not transfer to the quantum case where exact APSP requires $\wt{O}(n^{2})$ time but the quantum (exact) algorithm presented in this work outputs the diameter in $\wt{O}(n^{1.5})$ steps on sparse graphs. 
The area of classical approximation algorithms is currently heavily investigated. For state-of-the-art classical algorithms, the size of the gap varies depending on the quality of approximation, but usually, its size is close to $O(\sqrt{n})$. This is also what we witness in the quantum case where a $\Omega(n)$ lower bound stands against the $\wt{O}(n^{1.5})$ state-of-the-art upper bound for sparse graphs.
  
For a decent overview of classical lower bounds and comparisons of different approximation algorithms, the reader is referred to~\cite{backurs2021tightapproximationboundsgraph}.

There does not seem as if there could be much room for improvement over this runtime, at least in the context of the known quantum techniques. This is because we exhaust the (optimal) quadratic speedup in time complexity of the quantum search over $n$ vertices which intuitively is inherently necessary. It is also known due to D{\"u}rr et al.~\cite{durr2006quantum} that a quantum lower bound for the SSSP problem is $\Omega(\sqrt{nm})$ in the adjacency list model. This does not preclude other approaches from having a lower complexity, but it would be either via not considering every vertex or forgoing the capability of computing all paths and thereby most likely forgoing access to the information about lengths of these paths. \cref{alg:cap} finds the exact value of the diameter (or radius) in $\wt{O}(n\sqrt{m})$ time with high probability. For sparse graphs, this corresponds to the runtime of $\wt{O}(n^{1.5})$ which is a quadratic speedup over the $O(n^2)$ classical exact algorithm. Similarly, for dense graphs, the quantum algorithm has the complexity of $\wt{O}(n^2)$, which is compared to the classical $O(n^3)$ algorithm. Up to \textit{polylogarithmic} factors, on sparse graphs, the complexity of the exact quantum algorithm is equal to the classical approximation complexity of $\wt{O}(n^{1.5})$. We show later that a quantum lower bound for all of the considered problems is equal to $\Omega(\sqrt{nm})$ which on sparse graphs yields $\Omega(n)$, but no quantum algorithms are known that would have a runtime lower than $O(n^{1.5})$. Similarly, on dense graphs, the quadratic complexity gap remains as the lower bound is $\Omega(n^{1.5})$ and the upper bound is $\wt{O}(n^2)$. We leave the existence of that gap as one of the major and interesting open problems. 

\subsection{Lower Bounds}

\texttt{Parity} has been established to have a lower bound of $\Omega(n)$~\cite{durr2006quantum}. We can obtain a simple unconditional lower bound of $\Omega(n)$ for \texttt{Diameter}, \texttt{Radius} and \texttt{Eccentricity} by making use of the reduction from \texttt{Parity} given by D{\"u}rr et al.~\cite{durr2006quantum}.

An attempt to improve the lower bound for \texttt{Diameter} can be carried out by reducing from a problem with a higher lower bound. One such problem is \texttt{Triangle Collection}, with a classical lower bound of $\Omega(n^3)$ and a quantum lower bound of $\Omega(n^{1.5})$~\cite{ambainis2022matching}. However, in a quantum setting due to the size of the reduction the lower bound on \texttt{Triangle Collection}~\cite{dahlgaard2016hardness} implies again the $\Omega(n)$ lower bound matching the bound obtained from \texttt{Parity}.

\begin{proof}{(\cref{thm:ul_LOWER})}.
Observe that solving \texttt{Eccentricity} for a given vertex $v$ comes down to finding $d$ elements (edges that form the path). Finding $d$ elements of different kinds in a set of $m$ elements requires $\Omega(\sqrt{dm})$.

We are going to show that if one could find the eccentricity of a vertex in a graph faster than $O(\sqrt{dm})$ then one could find $d$ distinct items in a database of $m$ items also faster than in time $O(\sqrt{dm})$. In a general graph the eccentricity $d$ can be as large as $n$, giving us the general instance lower bound of $\Omega(\sqrt{nm})$.
The above statements hold for \texttt{Eccentricity} and \texttt{Diameter} (for the latter simply assume the $s$-$t$ eccentricity corresponds to the diameter of the graph), and the corresponding reduction is given in~\cref{fig:sparse}. For \texttt{Radius} however, the reduction from~\cref{fig:sparse}. does not apply as it could come down to finding only one element of a kind with the lowest weight. Nevertheless, simply by identifying $s$ with $t$, we are forming a "circle" with $d/2$ different "kinds" of edges, and solving \texttt{Radius} necessarily requires finding the minimum in each of the $d/2$ different kinds of edges (a "kind" here simply refers to a group of edges that connect $v_i$ to $v_{i+1}$); the lower bound follows. 

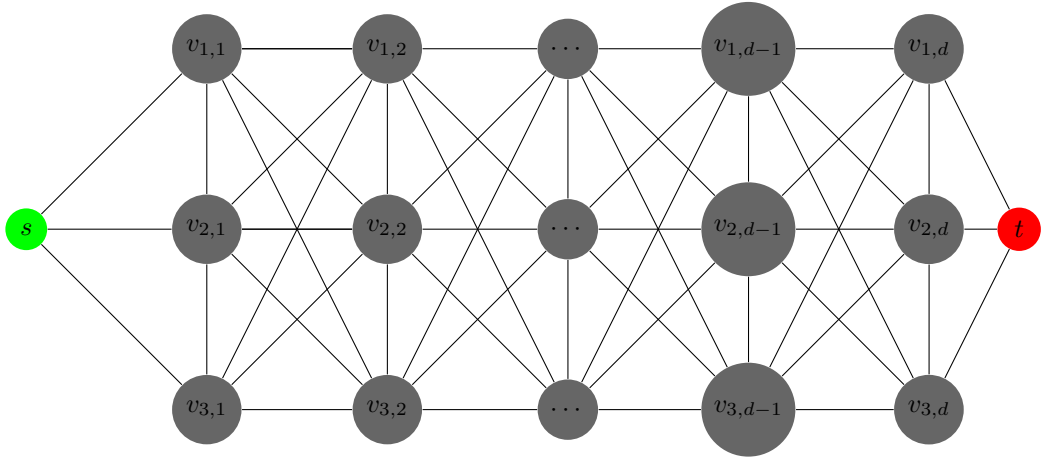
\begin{figure}[htbp]
    \centering
\begin{tikzpicture}
    [scale=1.2,auto=left,every node/.style={circle,fill=black!60}]

    % First row (y = 6)
    \node (n11) at (0,6)  {$v_{1,1}$};
    \node (n12) at (2,6)  {$v_{1,2}$};
    \node (n1k) at (6,6)  {$v_{1,d-1}$};
    \node (n1kp1) at (8,6)  {$v_{1,d}$};
    
    % Second row (y = 4)
    \node (n21) at (0,4)  {$v_{2,1}$};
    \node (n22) at (2,4)  {$v_{2,2}$};
    \node (n2k) at (6,4)  {$v_{2,d-1}$};
    \node (n2kp1) at (8,4)  {$v_{2,d}$};
    
    % Third row (y = 2)
    \node (n31) at (0,2)  {$v_{3,1}$};
    \node (n32) at (2,2)  {$v_{3,2}$};
    \node (n3k) at (6,2)  {$v_{3,d-1}$};
    \node (n3kp1) at (8,2)  {$v_{3,d}$};
    
    % Three dots indicating multiple columns (between 2nd and kth column)
    \node (dots1) at (4,6) {$\cdots$};
    \node (dots2) at (4,4) {$\cdots$};
    \node (dots3) at (4,2) {$\cdots$};

    % New vertex on the left side at the middle row level (y = 4)
  
    \node[fill=green] (new) at (-2,4)  {$s$};
    \node[fill=red] (old) at (9,4)  {$t$};
    % Edges connecting the new vertex to the first column
    \draw (new) -- (n11);
    \draw (new) -- (n21);
    \draw (new) -- (n31);
    \draw (old) -- (n3kp1);
    \draw (old) -- (n2kp1);
    \draw (old) -- (n1kp1);
    % Vertical edges between vertices in the same column
    \draw (n11) -- (n21);  % 1st column, 1st row to 2nd row
     \draw (n11) -- (n32);
    \draw (n21) -- (n32);  % 1st column, 1st row to 2nd row% 1st column, 1st row to 2nd row
    \draw (n21) -- (n31);  % 1st column, 2nd row to 3rd row
    \draw (n12) -- (n22);  % 2nd column, 1st row to 2nd row
    \draw (n22) -- (n32);  % 2nd column, 2nd row to 3rd row
    \draw (n1k) -- (n2k);  % k-th column, 1st row to 2nd row
    \draw (n2k) -- (n3k);  % k-th column, 2nd row to 3rd row

    % Vertical edges in the (k+1)-th column
    \draw (n1kp1) -- (n2kp1);  % k+1-th column, 1st row to 2nd row
    \draw (n2kp1) -- (n3kp1);  % k+1-th column, 2nd row to 3rd row

    % Connecting every vertex in column 1 to every vertex in column 2
    \foreach \i in {1,2,3} {
        \draw (n\i1) -- (n\i2);  % Connect each vertex in column 1 to each vertex in column 2
        \draw (n\i1) -- (n22);   % Connect column 1 to the second vertex in column 2
        \draw (n\i1) -- (n12);   % Connect column 1 to the first vertex in column 2
    }

    % Edges connecting every vertex in column 2 to every "3 dot" vertex
    \foreach \i in {1,2,3} {
        \draw (n\i2) -- (dots1);
        \draw (n\i2) -- (dots2);
        \draw (n\i2) -- (dots3);  % Connect column 2 to "3 dots"
    }

    % Edges connecting every "3 dot" vertex to every vertex in column k
    \foreach \i in {1,2,3} {
        \draw (dots1) -- (n\i k);
        \draw (dots2) -- (n\i k);
        \draw (dots3) -- (n\i k);  % Connect "3 dots" to column k
    }

    % Edges connecting every vertex in column k to every vertex in column k+1
    \foreach \i in {1,2,3} {
        \foreach \j in {1,2,3} {
            \draw (n\i k) -- (n\j kp1);  % Connect vertices in column k to column k+1
        }
    }
    \draw (dots2) -- (dots3);
    \draw (dots1) -- (dots2);

\end{tikzpicture}

\caption{{\bf Reduction for sparse graphs.} The graph used for the reduction from Minima Finding ($d$ elements of different types) to \texttt{Eccentricity}. There are $3d+2= n$ nodes, and a total of $m=10d-2$ edges. Vertices in each column are connected by edges of weight 0 to adjacent vertices in the same column, and with edges of some positive weights $w_{ij}$ between vertices in subsequent columns. Edges from $s$ and $t$ are all weights $0$. Finding the eccentricity of vertex $s$ comes down to finding a minimum weight edge between subsequent ``columns'', required to arrive at vertex $t$. The 0 weight edges between vertices in the same column ensure one can pick truly minimum weight edge between subsequent columns not the minimum weight edge from the vertex one happens to arrive at via the minimum weight edge from the previous column. Since edge weights are positive one has a guarantee that to find the eccentricity of $s$ one has to traverse all columns up to vertices in column $d$. Additionally, we require that in each subsequent column the edge weights connecting to the next column are smaller than all individual edge weights connecting to the current vertex. This requirement prevents a situation of moving back to the previous column of vertices, which could happen if suddenly all weights of edges connecting to the next (right) column are way larger than the weights of edges connecting to the current vertex from the previous (left) column. The value of a constant $3$ (the width of a layer) is an arbitrary choice to make the illustration simple but informative.}
\label{fig:sparse}
\end{figure}

The construction in~\cref{fig:sparse} works for sparse graphs and makes sure that the reduction is well-defined. It is clear that classically one has to explore all edges to find the eccentricity of $s$. If \texttt{Eccentricity} could be solved faster than $O(\sqrt{nm})$ then Minima Finding can be solved faster than $O(\sqrt{nm})$, which would contradict the Theorem 6 from~\cite[Theorem 6]{durr2006quantum} stating that the latter has an unconditional query lower bound of $\Omega(\sqrt{nm})$. It follows that one needs at least $\Omega(\sqrt{nm})=\Omega(n)$ to solve \texttt{Eccentricity} on weighted, sparse, undirected graphs. By the same argument, it follows that \texttt{Diameter} has  $\Omega(\sqrt{nm})=\Omega(n)$ and via identifying vertex $s$ with vertex $t$ the same lower bound follows for \texttt{Radius}. We include it as a demonstration of a new approach to proving the linear lower bound. Now we will show that we expand upon this reduction from quantum minima finding for dense graphs (see \cref{fig:dense}).

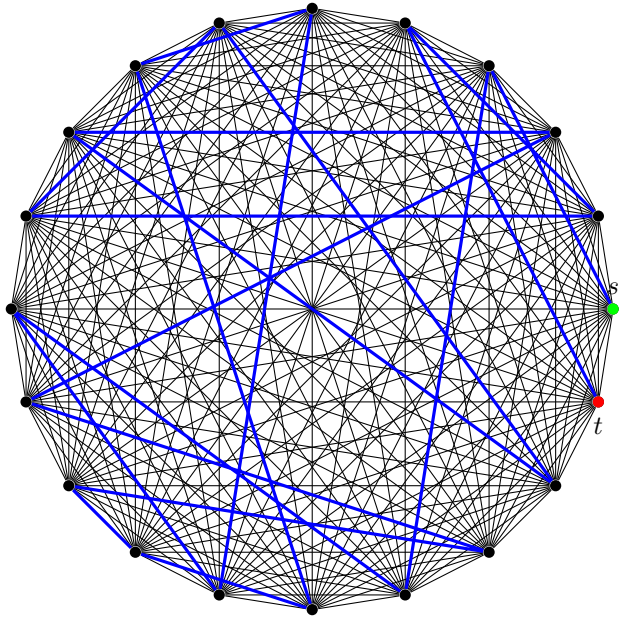
\begin{figure}[htbp]
    \centering
    \begin{tikzpicture}
    % Number of nodes
    \def\n{20}
    \def\radius{4cm}

    % Place nodes in a circular layout
    \foreach \i in {1, ..., \n} {
        \node[circle, fill=black, inner sep=1.5pt] (v\i) at ({360/\n * (\i - 1)}:\radius) {};
    }

    % Label the start and end of the path
   \node[circle, fill=green, inner sep=1.5pt, label=above:$s$] (v1) at ({360/\n * (1 - 1)}:\radius) {};
    \node[circle, fill=red, inner sep=1.5pt, label=below:$t$] (v\n) at ({360/\n * (\n - 1)}:\radius) {};

    % Draw all edges for the complete graph
    \foreach \i in {1, ..., \n} {
        \foreach \j in {1, ..., \n} {
            \ifnum\j>\i
                \draw (v\i) -- (v\j);
            \fi
        }
    }

    % Highlight the Hamiltonian path with blue edges from s to t

        \draw[blue, very thick] (v1) -- (v4);
        \draw[blue, very thick] (v4) -- (v17);
        \draw[blue, very thick] (v17) -- (v11);
        \draw[blue, very thick] (v11) -- (v15);
        \draw[blue, very thick] (v15) -- (v6);
        \draw[blue, very thick] (v6) -- (v8);
        \draw[blue, very thick] (v8) -- (v16);
        \draw[blue, very thick] (v16) -- (v14);
        \draw[blue, very thick] (v14) -- (v13);
        \draw[blue, very thick] (v13) -- (v18);
        \draw[blue, very thick] (v18) -- (v12);
        \draw[blue, very thick] (v12) -- (v3);
        \draw[blue, very thick] (v3) -- (v9);
        \draw[blue, very thick] (v9) -- (v19);
        \draw[blue, very thick] (v19) -- (v7);
    
        \draw[blue, very thick] (v7) -- (v10);
        \draw[blue, very thick] (v10) -- (v2);
        \draw[blue, very thick] (v2) -- (v5);
        \draw[blue, very thick] (v5) -- (v20);
        
\end{tikzpicture}
    \caption{{\bf Reduction for dense graphs.} The figure depicts a dense weighted graph, with the blue path being the lowest weight path from $s$ (green node) to vertex $t$ (red node) corresponding to the eccentricity of $s$.}
    \label{fig:dense}
\end{figure}

For dense graphs consider a complete weighted graph $G(E, V, w)$. On that graph, an eccentricity of some vertex $s$ in the worst-case scenario requires traversing all nodes $v \in V$ and finding $n-1$ edges with minimal weight. To see that, consider the starting vertex $s$ and the final (unknown) vertex $t$, to which the shortest path yields the eccentricity of $s$. Now if the edge connecting $s$ to $t$ is the minimal-weight edge incident on $s$ then just one iteration of search solves the problem. This is the best-case scenario, in the worst-case scenario the min weight edge incident on $s$ does not go to $t$. Instead, it leads to some vertex $v_1$, from which again the min weight edge does not go to $t$ (and the total weight of the path built so far is less than of all alternatives), the process can in the worst case scenario require to traverse all nodes before reaching $t$. An even easier example would be to consider the weights of all but one edge from each vertex to be very large. Then, in the worst case, these small-weight edges form a path from $s$ to $t$ that visits all nodes. This means that to find the shortest path from $s$ to $t$ one has to go through $n-1$ edges. There is an assurance that this kind of path will be the eccentricity of $s$, i.e. that overall it will be the longest of shortest paths from $s$. It is easy to see that the shortest $s$-$t$ path will necessarily involve the shortest $s$-$v_i$ paths.

Via the same argument as in the sparse graph case, in the dense case the same bound also applies to \texttt{Diameter} and \texttt{Radius}. 
\end{proof}

Our reduction improves the lower bound for all of the mentioned problems to $\Omega(\sqrt{nm})$ for (weighted, undirected graphs). In the case of dense graphs, this improves the previous, best lower bound of $\Omega(n)$.
We have considered reductions to \texttt{Diameter} that yield the linear lower bound. Reductions from \texttt{Triangle Collection}~\cite{dahlgaard2016hardness}, and \texttt{Parity} yield the $\Omega(n)$.

It follows the complexity gap between upper and lower bounds for \texttt{Eccentricity} is polylogarithmically small. Since \texttt{Diameter} and \texttt{Radius} are believed to be harder than \texttt{Eccentricity}, which makes it plausible that the upper bounds are going to be actually optimal (at least up to polylogarithmic factors), and the lower bounds need further refinements. 

\section{Approximation of diameter}

In this section, we consider approximating the diameter of an undirected graph. More specifically, we are aiming at outputing an estimate $\wh{\mathcal{D}}$ of diameter $\mathcal{D}$ such that $\frac{2}{3}\mathcal{D} \leq \wh{\mathcal{D}} \leq \mathcal{D}$.

Our algorithm is based on the previous work~\cite{aingworth1999fast,roditty2013fast}. They proposed the first deterministic and randomized algorithms to approximate diameter within the ratio $2/3$ without using matrix multiplication. Our quantum algorithm leverages quantum tricks to speedup some crucial subroutines in their algorithm which result in the polynomial speedup. In this section, we starts from brief review about classical results in~\cite{aingworth1999fast,roditty2013fast} along with defining some useful concepts. Then we introduce our key quantum subroutines Quantum Partial Breadth First Search. At last, our quantum algorithm comes up with a speedup in running time.

For simplicity, we consider undirected, unweighted graph $G=(V, E)$ here, which is a special case when $w(u,v)$ is $1$ if there exists an edge $(u,v) \in E$ and is $\infty$ otherwise. It's easy to generalize for directed and weighted graphs by replacing BFS with the famous algorithm for Single Source Shortest Path. Without a special statement, we use the link model by default.

\subsection{Breif review of the classical algorithm for approximating diameter}

Before introducing the well-known classical algorithm. We define what is the hitting set.

% \begin{definition}[Dominating Set]
%     Given a graph $G=(V, E)$, we call a set of vertices $A$ \textbf{dominating set} if $A \cap \BFS(v,1) \neq \emptyset$ for any $v \in V$.
% \end{definition}

% In other words, every vertex of the graph $G$ is either a neighbor of some node from the dominating set $A$ or is in the vertex set $A$.
        
\begin{definition}[Hitting set]
    A vertex subset $H_s$ of $V$ is an \textit{$s$-hitting set} if $H_s \cap N_s(v) \neq \emptyset$ for any $v \in V$.
\end{definition}

The $s$-hitting set is the set of vertices that every vertex from the graph $G$ can reach some node in the set $s$-hitting set within $s$ steps.

% \begin{lemma}[Sampling the hitting set~\cite{roditty2013fast}]
%     Sample a subset $H_s$ of $V$ with size $\Theta(n\log{n}/s)$ randomly. $H_s$ is $s$-Hitting Set of $G$ with high probability.
% \end{lemma}

In~\cite{aingworth1999fast}, the authors proposed a deterministic algorithm to distinguish graphs of diameter $2$ and $4$ in time $O(m\sqrt{n\log n})$ and extended to obtaining \cref{alg:classical_2vs3} to approximate diameter within ratio $2/3$ in time $O(m\sqrt{n \log n}+n^2\log n)$ for undirected and unweighted graphs. Moreover, their results can also be made to work for directed weighted graphs with arbitrary nonnegative weights with almost the same running time and approximation ratio. But in this work, we focus on the unweighted case. For completeness, the algorithm is given in~\cref{alg:classical_2vs3}.

\begin{algorithm}[htbp]
\caption{Classical algorithm to approximate diameter within $2/3$ ratio}
\label{alg:classical_2vs3}
\begin{algorithmic}[1]
    \Require A graph $G=(V,E)$.
    \Ensure $2/3\mathcal{D} \leq \wh{\mathcal{D}} \leq \mathcal{D}$.
    \State Compute a \textit{$s$-Partial-BFS tree} $BFS_s(v)$ for every vertex $v \in V$.
    \State Select the vertex $w=\arg\max_{v \in V}h_s(v)$.
    \State Compute $BFS(w)$ for $w$ and a BFS tree $BFS(u)$ for each vertex $u \in N_s(w)$.
    \State Compute the hitting set $H_s$ of $G$.
    \State Compute BFS tree from every vertex in $H$.
    \State Compute the BFS tree $BFS(u)$ for each vertex $u \in H$.
    \State Return estimator $\wh{\mathcal{D}}$ equal to the maximum depth of all $BFS$ trees for Step 3 and Step 6.
\end{algorithmic}
\end{algorithm}

\begin{theorem}[Approximate the diameter within ratio $2/3$~\cite{aingworth1999fast}]\label{thm:classical_approx}
    Given a undirected, unweighted graph $G=(V,E)$, \cref{alg:classical_2vs3} output the estimator $\wh{\mathcal{D}}$ s.t. $\frac{2}{3}\mathcal{D} \leq \wh{\mathcal{D}} \leq \mathcal{D}$ in $\wt{O}(ns^2+mns^{-1}+ms)$ time. Let $s=\sqrt{n}$, the running time is $O(m\sqrt{n \log n}+n^2\log^2 n)$.
\end{theorem}

% \begin{theorem}[Classical Approximation Algorithm\cite{aingworth1999fast}\cite{roditty2013fast}]
%     Algorithm \ref{classical 2 or 3 approx algorithm} output the estimator $\wh{D}$ s.t. $2/3D \leq \wh{D} \leq D$ in $\wt{O}(ns^2+mns^{-1}+ms)$ time. Let $s=\sqrt{n}$, running time is $\wt{O}(m\sqrt{n}+n^2\log^2{n})$
% \end{theorem}

The following work by~\cite{roditty2013fast} realizes that $ns^2$ term can be get rid of. The term of $ns^2$ comes from computing the partial BFS tree $BFS_s(v)$ for all $v \in V$. The tasks to accomplish based on this step are i) finding the deepest partial BFS tree $BFS_s(w)$ and ii) computing the hitting set $H_S$ later. Therefore, they modify the first step completely. More specifically, in~\cite{roditty2013fast}, they do not find the deepest partial BFS tree explicitly, but pick a different type of vertex to play the role of $w$. Then making the second task above fast can be accomplished easily with randomization. 
\begin{lemma}[Sampling hitting set~\cite{roditty2013fast}]\label{lem:samp_hitting_set}
    Sample a subset $H_s$ of $V$ with size $\Theta(n\log{n}/s)$ randomly. $H_s$ is $s$-Hitting Set of $G$ with high probability.
\end{lemma}
Combined with lemma~\cref{lem:samp_hitting_set}, their algorithm is as follows.
\begin{theorem}
    Given a undirected, unweighted graph $G=(V,E)$, \cref{alg:classical_2vs3} output the estimator $\wh{\mathcal{D}}$ s.t. $\frac{2}{3}\mathcal{D} \leq \wh{\mathcal{D}} \leq \mathcal{D}$ in $\wt{O}(mns^{-1}+ms)$ time. Let $s=\sqrt{n}$, the running time is $\wt{O}(m\sqrt{n})$.
\end{theorem}

Our quantum algorithm is basically based on the algorithm in~\cite{roditty2013fast}, we leave it to the next subsection. See~\cref{alg:quantum_2vs3}. The weighted cases of these theorems are almost the same. 

\subsection{Quantum algorithm for approximating diameter}
    
As one important step of the classical algorithm, partial BFS can be also sped up by using the ``small counting loophole'' of the quantum algorithm by using~\cref{thm:qtf}. Here we introduce our quantum partial BFS algorithm. The algorithm is based on the following lemma, which shows almost quadratic speedup to find the indices of all the neighbors.

\begin{lemma}
\label{lem:partial-search}
    Given a graph $G=(V, E)$ in the list model and a node $v \in V$. Suppose that the number of unvisited neighbors of $v$ is $m_v$ and $m_v \leq \deg(v)$. We can find the indices of $m_v$ unvisited nodes in $O(\sqrt{\deg(v)m_v}\log{\deg(v)})$ time.
\end{lemma}

\begin{proof}
    As $m_v$ is unknown, we use a series of numbers $2^0,2^1,\dots$ to approach $m_v$, until an exponent $k$ s.t. $2^{k-1} \leq m \leq 2^k$ is founded. In the algorithm, we run the quantum threshold find~\cref{thm:qtf} for at most $k$ times iteratively such like $\mathtt{QTF}(\deg(v),2^0)$, $\mathtt{QTF}(\deg(v),2^1)$, $\ldots$, $\mathtt{QTF}(\deg(v),2^k)$. To check if the node is visited or not and mark the visited nodes, we use the quantum history-independent data structure in~\cite{ambainis2007quantum}. Then we can un The last subroutine can tell us all $m_u$ neighbors need to traverse. The running time of such subroutine is $\sum_{i=0}^{\lceil{\log{m}}\rceil}O(\sqrt{\deg(v)2^i} \cdot \log \deg(v))=O(\sqrt{\deg(v)m_v}\log{\deg(v)})$. To check if the nodes are marked, we need to the quantum history-independent data structure from~\cite{ambainis2007quantum}.
\end{proof}

\begin{algorithm}[htbp]
\caption{Quantum Partial BFS algorithm $QPBFS(G,v_0,s)$}
\label{alg:q-partial-BFS}
\hspace*{\algorithmicindent} \textbf{Input:} An undirected, unweighted graph $G=(V,E)$, a starting node $v_0$ and a partial $s$.\\
\hspace*{\algorithmicindent} \textbf{Output:} The depth of BFS tree.
\begin{algorithmic}[1]
\State Intialize set $S=\emptyset$, $T=\emptyset$, $h=0$ and $v=v_0$.
\label{line:first_round_begin}
\State $S \gets S \cup \{v_0\}$
\State $T \gets T \cup \{(v,1)\}$, for all $ v \in N(v_0)$
\label{line:first_round_end}
\While{$T \neq \emptyset$ and $|S| < s$}
\label{line:pick_node_begin}
    \Repeat 
    \State Pick a node $(u,h)$ from $T$.
    \State $T \gets T/\{(u,h)\}$.
    \Until $T = \emptyset$ or $u \notin S$.
    \State Initialize $k \gets 0$.
    \State Initialize $\mathrm{INDEX\_SET} \gets \emptyset$.
    \While{$\mathrm{FLAG}=FALSE$}
        \State $\mathrm{FLAG}, \mathrm{INDEX\_SET} \gets \mathtt{QTF}(g_S \circ f_u(i), 2^k)$.
    \EndWhile
    \State $S \gets S \cup \mathrm{INDEX\_SET}$
    \State $T \gets T \cup \{(v,h+1)|v \in \mathrm{INDEX\_SET}\}$
\EndWhile\label{line:pick_node_end}
\end{algorithmic}
\end{algorithm}

\begin{theorem}[Quantum partial breadth-first search]
\label{thm:q-partial-BFS}
    Given a graph $G=(V,E)$ in the list model, a node $v \in V$ and a parameter $s$, \cref{alg:q-partial-BFS} can return $N_s(v)$ and depth of the $BFS_s(v)$ in $O(s^{3/2}\log{s})$ time.
\end{theorem}

\begin{proof}
    Our proof is based on~\cite[Theorem 18]{durr2006quantum}, but has a more refined analysis. \cref{alg:q-partial-BFS} can output at least $s$ closest neighbors of $v$ and the depth of BFS tree. The correctness follows. As for running time, in Step \ref{line:first_round_begin} to \ref{line:first_round_end}, we add all the neighbors of $v$ into the $S$, which takes $O(m_v)$ time.
    Suppose the ``while-loop'' Step \ref{line:pick_node_begin} to \ref{line:pick_node_end} has only $r$ rounds, and the vertices we picked after ``repeat loop'' are $v_1,v_2,\dots,v_r$ and every round $i=1,\dots,r$ the number of vertices we add to $S$ are $m_1,\dots,m_r$. W.l.o.g. assume that $\deg(v_i) \leq s$ for $i \in [r]$. Otherwise, the algorithm has already stopped. Then according to \cref{lem:partial-search}, we can find $\sum_{i=1}^r{m_i}$ nodes in $\sum_{i=1}^{r}O(\sqrt{\deg(v_i)m_i}\log{\deg(v_i)})$ time, which is bounded as follows
    \[
        \sum_{i=1}^{r}O(\sqrt{\deg(v_i)m_i}\log{\deg(v_i)}) \leq \sqrt{(\sum_{i=1}^{r}\deg(v_i)(\sum_{i=1}^{r}m_i)}\log{s} \leq \sqrt{s^2 \cdot s}\log{s}=O(s^{3/2}\log{s}),
    \]
    where the first inequality follows Cauchy--Schwarz inequality and observation that $\deg(v_i) \leq s$ for all $v_i \in V$ and the second inequality is from $r \leq s$ and $\sum_{i=1}^{r}m_i=s$.
\end{proof}

Now we are ready to present our quantum algorithm to approximate diameter within $2/3$ ratio. The algorithm for the case of weighted graphs is in~\cref{alg:quantum_2vs3}.

Besides the partial quantum breadth-first search. To traverse the whole graph, we have the following results, which are also the basis of our quantum speedup. The first is the quantum single source shortest path (QSSSP) from~\cref{thm:alg-qsssp}.
    
% \begin{theorem}[Quantum single source shortest path algorithm~\cite{durr2006quantum}]
%     The bounded error time complexity of computing single source shortest paths in the adjacent list model is $O(\sqrt{mn}\log^2{n})$.
%     % and $O(n^{3/2}\log^2{n})$ in the matrix model.
% \end{theorem}

% It was noted that, when the instance of the graph is an unweighted graph, Quantum SSSP can be regarded as a quantum BFS algorithm.

% \begin{corollary}[Quantum breadth-first search]
%     The bounded error time complexity of Breadth First Search in the array model is $O(\sqrt{mn}\log^2{n})$ and $O(n^{3/2}\log^2{n})$ in the matrix model. The algorithm is denoted by $\mathtt{QBFS}(G)$.
% \end{corollary}
    
% With a little modification of the algorithm by \cref{thm:alg-qsssp}, the algorithm to find $k$-nearest nodes in a directed graph can be solved in $O(\sqrt{mk}\log^2{n})$ time quantumly.

% \begin{corollary}[Quantum single source first $k$ shortest path algorithm]
%     The bounded error time complexity of single source $k$-shortest path in the array model is in $O(\sqrt{mk}\log^2{n})$ time.
% \label{qssksp}
% \end{corollary}

% The partial BFS can be sped up by leveraging the ``small counting loophole'' of the quantum algorithms.

% The quantum algorithm is the following:
\begin{algorithm}[H]
\caption{$\frac{2}{3}$-approximation quantum algorithm}
\label{alg:quantum_2vs3}
\hspace*{\algorithmicindent} \textbf{Input:} An undirected graph $G$ and a partial number $s$.\\
\hspace*{\algorithmicindent} \textbf{Output:} An estimate $\wh{ \mathcal D}$ of diamter  s.t. $\frac{2}{3} \mathcal D \leq \wh{ \mathcal D} \leq \mathcal D$. 
\begin{algorithmic}[1]
    \State Sampling $\Theta(ns^{-1}\log{n})$ nodes from $V$ to form set $H_s$. \Comment{Sampling a $s$-hitting set $\mathrm{H_s}$} \label{line:samp_hitting_set}
    \State $(w, h_s(w)) \gets \mathtt{QMF}(\{\mathtt{QPBFS}(G, v_i, s).second\}_{v_i \in \mathrm{H_s}})$ where every $v_i \in \mathrm{H_s}$.
        \label{line:search_w}
    \State $(N_s(w)(w),h_s(w)) \gets \mathtt{QPBFS}(G, w, s)$
        \label{line:find_N_s(w)}
    \State $(\_\_, \wh{ \mathcal D}) = \mathtt{QMF}(\{\mathtt{QSSSP}(G,v_i).second \}_{v_i \in H_s \cup N_s(w)})$
    \label{line:approx_diameter}
    \State Output $\wh{ \mathcal D}$.
\end{algorithmic}
\end{algorithm}

\begin{theorem}[Quantum algorithm to approximate diameter]
\label{thm:q_approx_alg}
    Given an unweighted graph $G=(V, E, w)$, \cref{alg:quantum_2vs3} can approximate $\mathcal{D}$ by $\wh{D}$ s.t. $2/3D \leq \wh{D} \leq D$ . And the algorithm runs in $\wt{O}(\sqrt{m}n^{3/4})$ time.
\end{theorem}

\begin{proof}
Correctness follows the argument of ~\cite{aingworth1999fast} and~\cite{roditty2013fast}. For the running time, step \ref{line:samp_hitting_set} samples $\Theta(ns^{-1}\log{n})$ vertecies to form hitting set $H_s$, which takes $O(ns^{-1}\log{n})$ time. Step \ref{line:search_w} applies Quantum Minimum Finding over the Quantum Partial BFS search. It takes $O(s^{3/2} \cdot \sqrt{n\log{n}/s}\log{(n\log{n}/s)})$ time. Step \ref{line:find_N_s(w)} takes $O(s^{3/2})$ which is dominated by running time of Step $2$. Each quantum SSSP takes $O(\sqrt{mn}\log^{5/2}{n})$ time. In step \ref{line:approx_diameter}, we use the quantum BFS algorithm, by~\cref{thm:q-partial-BFS} substitute quantum PBFS in the last step and apply quantum minimum finding over set $H_s \cup N_s(w)$. So it takes $O(\sqrt{n\log{n}/s+s}\sqrt{mn}\log^{5/2}{n})$ time. Overall, the total running time is 
\[
O(s^{3/2} \cdot \sqrt{n\log{n}/s}\log{(n\log{n})} + (\sqrt{s+n\log{n}/s})\sqrt{mn}\log^{5/2}{n}) = \wt{O}(\sqrt{mn(n/s+s)}).
\]
let $s=n^{1/4}$, the running time is $\wt{O}(n^{5/4})$.
\end{proof}

\begin{corollary}
When graph is sparse with $m=O(n)$, the expected running time of~\cref{alg:quantum_2vs3} is $\wt{O}(n^{5/4})$. When graph is dense with $m=O(n^2)$, the expected running time of~\cref{alg:quantum_2vs3} is $\wt{O}(n^{7/4})$.
\end{corollary}

% \begin{corollary}
% When graph is sparse with $m=O(n)$, let $s=n^{1/4}$, the expected running time is $O(n^{1.25})$. When graph is dense with $m=O(n^2)$, the expected running time is $O(n^{1.75})$ by setting $s=n^{1/4}$.
% \end{corollary}

For unweighted cases, we can simply replace the quantum SSSP algorithm with the quantum BFS algorithm (\cref{thm:qbfs}), the $\sqrt{mn}$ factor induced by quantum SSSP will be replaced by $n$. And resulting in the following corollary.

\begin{corollary}
    Given an unweighted graph $G=(V,E)$, there exists quantum algorithm to approximate $\mathcal{D}$ by $\wh{\mathcal{D}}$ s.t. $2/3\mathcal{D} \leq \wh{\mathcal{D}} \leq \mathcal{D}$ in $\wt{O}(n^{5/4})$ time.
\end{corollary}

\section*{Acknowledgement}

J.B. is supported by Quantum Advantage Pathfinder (QAP) project of UKRI Engineering and Physical Sciences Research Council (Grant No. EP/X026167/1).

\bibliographystyle{alpha}
\bibliography{references}

\end{document}